\documentclass[conference]{IEEEtran}

\usepackage{algpseudocode} 
\usepackage{algorithm} 
\usepackage[utf8]{inputenc} 
\usepackage[T1]{fontenc}
\usepackage{url}
\usepackage{ifthen}
\usepackage{cite}
\usepackage[cmex10]{amsmath} 

\usepackage{cite}
\ifCLASSINFOpdf
 
\else
   \usepackage[dvips]{graphicx}
 \fi
\usepackage[cmex10]{amsmath}
\usepackage{amssymb}
\usepackage{stfloats}
\usepackage{tikz-cd}
\usepackage{amsthm} 
\usepackage{wrapfig}
\usepackage{amssymb}
\usepackage{latexsym}
\usepackage{bm, bbm} 
\usepackage{epsfig}
\usepackage{graphicx}
\usepackage{epsfig}
\usepackage{multicol}
\usepackage{psfrag}
\usepackage{float}
\usepackage{cite}
\usepackage{subfigure}
\usepackage{color}

\newtheorem{definition}{Definition}
\newtheorem{theorem}{Theorem}
\newtheorem{statement}{Statement}
\newtheorem{remark}{Remark}
\newtheorem{lemma}{Lemma}

\title{Cyclic Group Projection for Enumerating Quasi-Cyclic Codes Trapping Sets}

\author{%
  \IEEEauthorblockN{Vasiliy~Usatyuk}
  \IEEEauthorblockA{R\&D Department, T8\\ 
                    Email: L@lcrypto.com}
  \and
  \IEEEauthorblockN{Yury Kuznetsov}
  \IEEEauthorblockA{R\&D Department, T8\\ 
                     Email: K@lcrypto.com}
   \and
  \IEEEauthorblockN{Sergey Egorov}
  \IEEEauthorblockA{Computer Science Department, South-West State University\\ 
                     Email: sie58@mail.ru}

                     }

\begin{document}
\maketitle
\begin{abstract}
This paper introduces a novel approach to enumerate and assess Trapping sets in quasi-cyclic codes,  those with circulant sizes that are non-prime numbers. Leveraging the quasi-cyclic properties, the method employs a tabular technique to streamline the importance sampling step for estimating the pseudo-codeword weight of Trapping sets.  The presented methodology draws on the mathematical framework established in the provided theorem, which elucidates the behavior of projection and lifting transformations on pseudo-codewords.
   \end{abstract}\vspace{-1mm}

%

\newcommand{\F}{\mathbb{F}_2}
\newcommand{\Z}{\mathbb{Z}}
\newcommand{\R}{\mathbb{R}}
\newcommand{\supp}{\mathrm{supp}\,}
\renewcommand{\P}{\mathbb{P}}
\renewcommand{\mod}{\, \mathrm{mod} \,}
\newcommand{\Wcal}{\mathcal{W}}

\section{Introduction}

The incorporation of graph models as a theoretical framework has sparked a revolutionary impact across a wide array of domains, notably in areas such as forward error correction, source coding, compressed sensing, classical and quantum machine learning/computation, quasi-optimal electrical power and transport networks,  post-quantum cryptography.

In the publication \cite{Equilibrium}, a noteworthy connection has been revealed between sophisticated Deep Neural Network (DNN) architectures, exemplified by state-of-the-art models such as Mega and ChordMixer Transformers in Natural Language Processing, and the Generalized Irregular Repeat Accumulate and Cage-graph graph models. Delving into the intricacies of the approximate landscape of Bethe-Hessian (Bethe-permanent) and Quasi-Newton landscape Hessian (permanent) in deep neural networks goes beyond the conventional lazy regime, sparking inquiries into the symmetry and asymmetry of statistical manifolds within arbitrary non-linear channels \cite{lazy}. This exploration reignites discussions on the capacity of graph models under non-linear data channels in high noise environments, with a specific emphasis on the generalized covariance evolution describing linear size Trapping sets (TS) and low noise conditions for sublinear linear size TS. A notable discovery relates to certain nonlinear data channels, where the clarification of ground states (barycenter, \cite{Chavira07,Vuffray}) involves embedding tori and circular hyperboloids, revealing their equivalence to quasi-cyclic codes. The count of energy minima corresponds to the size of the circulant, underscoring the crucial role of quasi-cyclic codes in linearization, \cite{Equilibrium}. As a result, advancing neural network architectures and establishing highly reliable, length-adaptive error correction codes based on quasi-cyclic graphs necessitates the development of efficient algorithms aimed at enhancing the spectrum of TS.
 
\section{Basic definition}

To facilitate clarity in our exposition, we denote the binary field as $GF(2) \equiv \mathbb{F}_2$. Formulas employ a zero-based numbering convention. The cardinality of a finite set $M$ is denoted as $|M|$. In expressions where addition is conducted within the field $\F$, the symbol $\oplus$ is employed. The set of integers is denoted as $\Z$. Our foundational notation is introduced following the conventions outlined in~\cite{PolUsaVor}.
\begin{definition}
The circulant matrix \(Q\), denoted as \(\{Q_{ij}\}_{i,j=0}^{z-1}\) with \(z > 0\), is defined as follows:

\[
Q_{ij} = 
\begin{cases}
    1, & \text{if } i + 1 \equiv j \pmod{z}, \\
    0, & \text{otherwise},
\end{cases}
\]
for \(i, j = 0, \ldots, z-1\).

Additionally, it is assumed that the inverse of \(Q\), denoted as \(Q^{-1}\), is equivalent to the zero matrix \(\mathbf{0}\) in \(\F^{z \times z}\). We assume $Q^{-1} \equiv 0\in \F^{z\times z}$.
\end{definition}

\begin{definition}
The quasi-cyclic parity-check matrix, denoted as \(H\) and of size \(mz \times nz\) with \(m, n > 0\) and circulant \(z\), is expressed as:
\[ H = \left[ 
\begin{array}{llll}
Q^{a_{00}} & Q^{a_{01}} & \ldots & Q^{a_{0,n-1}} \\
Q^{a_{10}} & Q^{a_{11}} & \ldots & Q^{a_{1,n-1}} \\
\vdots & \vdots & \vdots & \vdots \\
Q^{a_{m-1,0}} & Q^{a_{m-1,1}} & \ldots & Q^{a_{m-1,n-1}} \\
\end{array}
\right]. \]
\end{definition}

\begin{definition}
The exponential matrix \(E(H)\), belonging to \(\Z^{m\times n}\) and defined for \(H\) as follows:
\[ E(H) = \left[ 
\begin{array}{llll}
a_{00} & a_{01} & \ldots & a_{0,n-1} \\
a_{10} & a_{11} & \ldots & a_{1,n-1} \\
\vdots & \vdots & \vdots & \vdots \\
a_{m-1,0} & a_{m-1,1} & \ldots & a_{m-1,n-1} \\
\end{array}
\right]. \]
\end{definition}

\begin{definition}
The mother matrix or base (protograph) matrix \(M(H)\) is obtained from \(E(H)\) by replacing -1 with 0, and all other values with 1.

\end{definition}

\begin{definition}
The linear code defined by a quasi-cyclic parity-check matrix \(H\), denoted as \(C = C(H)\), is a subset of \(\F^{nz}\) and is defined as:

\[ C(H) = \{x \in \F^{nz} : \quad Hx^T = 0\}. \]
\end{definition}

\section{Cyclic Group Decomposition}
Let's assume that the circulant size \(z\) is equal to the product of two integers \(l\) and \(z_\ast\), where both \(l\) and \(z_\ast\) are greater than 1.  Notably, this property is found in various practical instances, including length-adapted quasi-cyclic codes such as the enhanced mobile broadband (eMBB) 5G QC-LDPC codes, DVB-S2/X codes, Graph Neural Network and others.

\begin{definition}\label{defPP} ({\rm{\cite{PolUsaVor}},~3.1), \cite{MyungYang1,MyungYang2}}.
The mapping \(\P_{z\rightarrow z_\ast}\) from a set of quasi-cyclic matrices \(\F^{mz\times nz}\) with circulant size \(z\) to a set of quasi-cyclic matrices \(\F^{mz_\ast\times nz_\ast}\) with circulant size \(z_\ast\) is determined by the mapping of their exponential matrices, denoted as \(E_z\rightarrow E_{z_\ast}\). Here, \(E_z\) and \(E_{z_\ast}\) belong to \(\Z^{m\times n}\) and follow a modular lifting scheme:

\[
\{E_{z_\ast}\}_{ij} = 
\begin{cases}
    \{E_z\}_{ij} \mod z_\ast, & \text{if } \{E_z\}_{ij} \geqslant 0, \\
    -1, & \text{if } \{E_z\}_{ij} = -1.
\end{cases}
\]
\end{definition}
It's worth noting that the mapping \(\P_{z\rightarrow z_\ast}\) preserves the structure of the mother matrix. Specifically, the Tanner graph of a quasi-cyclic matrix \(\tilde{H}\) with the exponential matrix \(E_z = E(\tilde{H})\), along with the natural projection, constitutes an \(l\)-graph covering of the parity-check matrix \(H\) with \(E_{z_\ast} = E(H)\), \cite{Cover}.

\begin{definition}\label{defP}
Linear mapping $P_{z\rightarrow z_\ast}$ set of vectors $\F^{nz}$ to a set of vectors $\F^{nz_\ast}$
we will determine by the following procedure.
For arbitrary $x\in \F^{nz}$ let's denote $y = P_{z\rightarrow z_\ast} x \in \F^{nz_\ast}$.
Component $y_j$, ${0 \leqslant j < nz_\ast - 1}$, where $j = qz_\ast + r$, $0\leqslant r < z_\ast$, is calculated as follows:
\begin{equation}\label{eqA}
y_j = x_{qz+r}\oplus x_{qz + r + z_\ast} \oplus \ldots \oplus x_{qz + r + (l-1)z_\ast}.
\end{equation}
\end{definition}

If we consider \(x\) as the codeword of a quasi-cyclic parity-check matrix \(\tilde{H}\), then the vector \(P_{z\rightarrow z_\ast} x\) becomes the image of a pseudocode word (as defined in~\cite{Cover}) corresponding to \(x\) under the homomorphism \(\Z \rightarrow \Z_2\). This \(l\)-covering is determined by the mapping \(\P_{z\rightarrow z_\ast}\). Graph covers, accompanied by their corresponding lifting operations, can be seen as enhancing the delay operator in convolutional codes. The tail-bited  form of convolutional codes evolves into Quasi-Cyclic codes with an expanding circulant size, as elaborated in \cite{Tanner, Cover, Cover1, Cover2}. As observed, both codewords \cite{MacKayFirstBoundAndMixedAuth, RoxanaCodeBound, BPMETCodeBound} and pseudo-codewords \cite{RoxanaPseudoBound, VontobelPseudoBound,DeBanih} consistently approach their respective upper bounds. 


\begin{definition} Projection index operator: 
The function \(\pi_{z\rightarrow z_\ast}\) is defined as:

\[
\pi_{z\rightarrow z_\ast}: \{0, 1, 2, \ldots\} \longrightarrow \{0, 1, 2, \ldots \},
\]

and it operates as follows: for \(j \geqslant 0\), \(j = qz + r\), \(0 \leqslant r < z\),

\[
\pi_{z\rightarrow z_\ast}(j) = qz_\ast + (r \mod z_\ast).
\]

Using the definition of \(\pi_{z\rightarrow z_\ast}\), \(P_{z\rightarrow z_\ast}\) can be expressed as:

\[
\forall x \in \F^{nz} \quad \{P_{z\rightarrow z_\ast} x\}_j = \sum_{k\in \pi^{-1}(j)} x_k, \,\, j = 0, 1, \ldots, nz_\ast - 1,
\]

where the summation is performed in \(\F\).

\end{definition}

\begin{definition}~{\rm{\cite{PriceHall}[Definition~27]}}
A \((w, v)\)-TS pseudo-codeword is a vector \(x\in \F^{nz}\) satisfying the condition that the Hamming weights of the vectors \(x\) and \(Hx^T\) are equal to \(w\) and \(v\) respectively, where \(w > 0\) and \(v \geqslant 0\).
\end{definition}

We fix an arbitrary quasi-cyclic check matrix \(H\in \F^{mz\times nz}\) with a circulant size of \(z\), where \(z = lz_\ast\), \(l, z_\ast\in \Z\), and \(l, z_\ast > 1\).

\begin{theorem}\label{th1}
For any pseudocode word $(w, v)$-pseudocode word $x\in \F^{nz}$ relative to the check matrix $H$ 
vector $x_\ast = P_{z\rightarrow z_\ast}x\in \F^{nz_\ast}$ 
will $(w', v')$-TS pseudo-code word relative to the parity check matrix $H_\ast = \P_{z\rightarrow z_\ast}H$, 
and $w' \leqslant w, \; v'\leqslant v,$
$$
 \quad v-v'=2r',\, w - w'=2r'',\,\, r', r'' \in \{0, 1, 2, \ldots\}.
$$
As a consequence, the linear transformation
$$
P_{z\rightarrow z_\ast}: \F^{nz}\longrightarrow \F^{nz_\ast}
$$
maps the linear code $C(H)$ to some subset of code $C(H_\ast)$.
\end{theorem}

\begin{remark}
Case $w' = 0$ possible if $w$~-- even number. In this case, the statement of the theorem makes no sense. 
\end{remark}
\begin{remark}
The corollary indicated at the end of the theorem is directly deduced from the theorem ~4.4 from \cite{Cover}.
\end{remark}
\begin{proof}
To simplify the proof, we will use the symbols $\P$, $P$, $\pi$ without index $z\rightarrow z_\ast$.

The first part of the theorem is that the number of non-zero components of the vector
$x_\ast = Px$ will be the same as $x$, or less by an even number. This follows from the definition of the operation $P$~(\ref{eqA}).
To prove the second part, we introduce a number of notations and simplifications.

From the definition of the Hamming weight it follows that it is sufficient to prove the statement of the theorem for $m = 1$, that is for parity-check matrix
${H\in \F^{z\times nz}}$ from one block (circulant) row. 

For an arbitrary vector $y\in \F^q$ through $\supp y$ we will denote the set of indices $j\in \{0, 1, \ldots, q-1\}$ such that
$y_j = 1$.

$\sigma: \Z \longrightarrow \Z_2$~--- homomorphism of additive groups.

Let's assume that the index $i$ defines some scalar linear equation. Will write $\rho(i)=1$ if this equation does not hold for vector $x$ (clear from the context) and $\rho(i)=0$ otherwise.

Let us denote by $V^0$ and $V^0_\ast$ subsets of vertex indices $\{i\}_{i=0}^{ nz-1}$ and $\{i\}_{i=0}^{ nz_\ast-1}$
respectively, corresponding to non-zero circulants of the matrices $\F^{z\times nz}$ and $\F^{z_\ast\times nz_\ast}$.

Since any active parity-check~--- i.e. scalar equation of the system $Hu = 0$ ($H_\ast v = 0$), containing non-zero
vector component $x$ ($x_\ast$)~---  depends only on $\supp x \cap V^0$ ($\supp x_\ast \cap V^0_\ast$) and, in addition, for $\tilde{x}$, $\supp \tilde{x} = \supp x\cap V^0$ equality is satisfied
$$
\supp (P\tilde{x}) = \supp (Px) \cap V^0_\ast,
$$
we can assume that $\supp x\subset V^0$, $\supp x_\ast\subset V^0_\ast$. There are relations
$$
\pi(V^0) = V^0_\ast, \quad \supp x_\ast \subset \pi(\supp x)
$$

Quasi-cyclic matrix structure $H$ and $H_\ast$ uniquely defines mappings
$$
\begin{aligned}
c(\cdot): \quad & V^0 \longrightarrow \{0, 1, \ldots, z - 1\}\\
c_\ast(\cdot): \quad & V^0_\ast \longrightarrow \{0, 1, \ldots, z_\ast - 1\}
\end{aligned}
$$
corresponding to the vertex index $j$ row index $i$ such that the element of the corresponding check matrix
in position $(i, j)$ is equal to 1. Let us prove the equality
\begin{equation}\label{eq1}
c_\ast(\pi(j)) = \pi(c(j)), \quad j\in V^0.
\end{equation}
Let's fix an arbitrary index $j\in V^0$, $j = qz  + r$, $0 \leqslant q < n$, $0 \leqslant r < z$. 
Then $c(j)= (r - a_q)\mod z$, where $a_q$~--- exponent (upper index) $q$-th circulant of the matrix $H$, 
and the operation $a \mod z$ is defined so that the result lies in the interval $[0, z-1]$ for any integer $a$. The chain of equalities is fair
\begin{multline*}
\pi(c(j)) = ((r - a_q)\mod z) \mod z_\ast = \\
= (r - a_q)\mod z_\ast =  (r\mod z_\ast - a_q \mod z_\ast)\mod z_\ast =
  \\ = c_\ast(qz_\ast + r\mod z_\ast) = c_\ast(\pi(j)).
\end{multline*}

Mapping
$$
\pi: \{0, 1, \ldots, nz-1\}\longrightarrow \{0, 1, \ldots, nz_\ast - 1\}
$$
defines the decomposition of a set $V^0$ into disjoint subsets
$W_j := \pi^{-1}(j)$, $ j\in V^0_\ast$.

By definition of operation $P$ set of indexes $\supp x_\ast$ correspond one-to-one with these
$W_j$, for which $W_j \cap \supp x$ consists of an odd number of elements. We will denote this fact
$W_j \sim \supp x_\ast$.

\begin{lemma}\label{lem1}
Let's $\Wcal$~--- arbitrary family of sets $W_j$, $j\in V_\ast^0$. 
$$
 \hat{\Wcal} := \{ W\in \Wcal: \; W \sim \supp x_\ast \}.
$$
Then the following parities coincide:
\begin{equation}\label{lem1eq}
\sigma(|\hat{\Wcal}|) = \sigma\left(\left| \bigcup_{W\in \Wcal}\{W \cap \supp x\}\right|\right)
\end{equation}
\end{lemma}

\begin{proof}
Let us first assume the equality $\hat{\Wcal} =\Wcal$. Considering that the sets in the union on the right side ~(\ref{lem1eq}) do not intersect, we get:
\begin{multline*}
\left| \bigcup_{W\in \Wcal}\{W\cap \supp x\}\right | = \sum_{W\in \Wcal} | \{W\cap \supp x\} | = \\
= \sum_{W\in \Wcal} (1 + 2s_W) = |\Wcal|  + 2s,
\end{multline*}
where $s_W$, $s$~--- some integers. This proves the lemma in a special case. If  $\hat{\Wcal}\neq \Wcal$,
then the number of indices $\supp x$ in every set $\Wcal \setminus \hat{\Wcal}$ will be either zero or a positive even number, and we can refer to an already proven fact.
\end{proof}

Using the following notation $C^a$, $C^a_\ast$ let us denote the sets of active parity-checks for $H$, $H_\ast$ relative to vectors $x$, $x_\ast$:
$$
C^a = c(\supp x), \quad C^a_\ast = c_\ast(\supp x_\ast)
$$

\begin{lemma}\label{lem2}
For any $s\in C^a_\ast$ there is a non-empty set of indices ${\{i_1^s, \ldots, i_{k(s)}^s\}}$, $i_l^s\in C^a$, $l = 1, \ldots, k(s)$
such that
$$
\rho(s) = \rho(i_1^s)\oplus \ldots \oplus \rho(i^s_{k(s)}),
$$
and for different $s', s''$ the index sets do not intersect.
\end{lemma}

\begin{proof}
Let's fix the $s\in C^a_\ast$ and consider the equation
\begin{equation}\label{eq3}
s = c_\ast(\pi(j)), \quad j\in \supp x.
\end{equation}
Set of solutions to the ~(\ref{eq3}) denote by  $R = R(s)$. From the equality we see ~(\ref{eq1}) follows it follows that the following sets are either
fully contained in $R(s)$, or do not intersect with it:
\begin{align}
& \supp x \cap W_j,   && j\in V^0_\ast \label{eq4}\\
& \supp x \cap c^{-1}(i),  &&  i = 0, 1 , \ldots, z - 1 \label{eq5}
\end{align}
contained in  $R(s)$ non-empty representatives ~(\ref{eq4}), (\ref{eq5}) form two partitions of this set.
Let's define 
$$
\Wcal(s) := \{W_j: \, \supp x \cap W_j \subset R(s), \, j\in V_\ast^0\}.
$$
As a set $\{i_l^s\}_{l=1}^{k(s)}$ let's take different indices $i$, ${0 \leqslant i < z-1}$,
 for which the sets ~(\ref{eq5}) are non-empty and contained in $R(s)$. Number of solutions to the equation
 $c_\ast(j_\ast) = s$, $j_\ast \in \supp x_\ast$ equal to the number of such $W\in \Wcal(s)$, that $W \sim \supp x_\ast$.
Hence and from the lemma ~(\ref{lem1}) we get equalities:
\begin{multline*}
\rho(s) = \sigma\left( |\{W\in \Wcal(s):\, W\sim \supp x_\ast\}|\right) = \sigma(| R(s) |) = \\
= \sigma\left( \left | \bigcup_{l = 1}^{k(s)} \{ \supp x \cap c^{-1}(i_l^s)\}\right| \right ) =  \\
=  \sum_{l = 1}^{k(s)} \sigma( | \{ \supp x \cap c^{-1}(i_l^s)\} |) =  \sum_{l = 1}^{k(s)} \rho(i_l^s), 
\end{multline*}
where the summation is performed in $\F$.

If $s' \in C^a_\ast$, $s'\neq s$, that set of indexes $\{i_l^{s'}\}_{l=1}^{k(s')}$ cannot intersect with set $\{i_l^s\}_{l=1}^{k(s)}$ due to equality ~(\ref{eq1}).
\end{proof}
\begin{remark}
In certain scenarios, it is possible that there are no solutions \(j\) to the equation \(y_j = 1\), where \(j\) is projected into \(\text{supp } x_\ast\) according to the definition given in Equation (\ref{eq3}). By repeating the reasoning presented in the lemma, it becomes evident that in such cases, the number of active checks unsatisfied, denoted as \(\{i_l^s\}_{l=1}^{k(s)}\), will be even.
\end{remark}

The remaining part of the theorem can be derived from Lemma \ref{lem2}. For any unsatisfied check \(s \in C^a_\ast\), where \(\rho(s) = 1\), the number of unsatisfied checks among \(\{i_l^s\}_{l=1}^{k(s)}\subset C^a\) will either be equal to 1 or exceed 1 by an even number. Additionally, considering the remark, the checks from \(C^a\) that are not related to \(C^a_\ast\) will result in an even number.
\end{proof}

{\flushleft \textbf{Numerical Example 1.}}
Consider a QC parity-check matrix, with circulant $z=128$ and exponential matrix 
\begin{gather}
 E(H) = [E_1 | E_2]  , \, \text{where} \label{NumEx1.1} \\
\substack{E_1 =  \left[
\begin{array}{rrrrrrrrrr}
21&65	& 126& 39	& 84	& 89	&16		&89		&94		&28	  \\
12	& 29	& 27	& 105	& 0	      & 0		&0		&2		&108	&78	  \\
47	& 26	& 97	& 40	& 66	& 0		&81		&1		&33		&97	 \\
113	& 21	& 64	& 26	& 53	& 103	&85		&81		&102	&48
\end{array} \right]  \label{NumEx1.2} } \\
\substack{E_2 = \left[
\begin{array}{rrrrrrrrrr}
31		&12		&101	&9		&122	&99		&90		&65		&41		&-1\\
126	      &98		&27		&1		&41		&118	&49		&17		&12		&11\\
89		&37		&52		&42		&2		&120	&38		&-1		&86		&7\\
113	      &98		&71		&70		&121	&66		&95		&40		&-1		&23
\end{array} \right] } \label{NumEx1.3}
\end{gather} 
For $H$ a lifting operation, inverse to projection, has been performed. For this purpose the binary matrix was used
$B=\{b_{ij}\}$, $i=0,\ldots,3$ , $j = 0,\ldots,19$, whose elements are drawn from a pseudo-random binary distribution.
The resulting $\tilde{H}\in\F^{1024\times5120}$, $\tilde{z} = 2z = 256$ defined by its exponential matrix
$\tilde{E} = E(\tilde{H})$, 
$$
\{\tilde{E}\}_{ij} = \left \{ 
\begin{array}{ll}
\{E\}_{ij} + b_{ij}z, & \{E\}_{ij} \geqslant 0, \\
-1, & \{E\}_{ij} = -1, \\
\end{array}\right .
$$
where $E = E(H)$. It's obvious that $H = \P_{256\rightarrow 128}(\tilde{H})$.

It has been determined using an approximate method that there are 874903 \((w, v)\)-Trapping sets in the space \(\F^{5120}\) under the constraints \(w\leqslant 30\) and \(v\leqslant w\). Each of these words has undergone projection using the operation \(P_{256\rightarrow 128}\) (refer to Definition \ref{defP}). The resulting words have been classified into their respective classes \((w', v')\). The distribution of differences \((w-w', v-v')\) is presented in Table \ref{tabl1}.
Note that two class changes 
$$
 (w, v)\rightarrow (w, v), \quad (w, v)\rightarrow (w, v - 2)
$$
make up for this matrix and its lifting 99.26\% cases.

\begin{table}[t]
\caption{Distribution of differences $(w-w', v-v')$} 
\centering 
\begin{tabular}{c | c | c }
$w - w'$ & $v - v'$ & frequency, \% \\
\hline 
   0 &  0 & 93.071689 \\
   0 &  2 &  6.191543 \\
   0 &  4 &  0.191678 \\
   0 &  6 &  0.002743 \\
   2 &  0 &  0.360154 \\
   2 &  2 &  0.152474 \\
   2 &  4 &  0.016458 \\
   2 &  6 &  0.000571 \\
   4 &  0 &  0.006400 \\
   4 &  2 &  0.003543 \\
   4 &  4 &  0.001257 \\
   6 &  0 &  0.000685 \\
   6 &  2 &  0.000457 \\
   6 &  4 &  0.000228 \\
   8 &  4 &  0.000114 \\
\end{tabular}
\label{tabl1}
\end{table}

\section{QC Codes TS Enumerating  }\label{algorithm}

We will assume that we are given a quasi-cyclic parity-check matrix $H\in \F^{mz\times nz}$, 
$m, n,  z>0$. Automorphism group (~\cite{MacWilliams}[Ch.~8,~\S~5]) $Aut(C(H))$ linear code $C(H)$ contains a subgroup $\mathcal{G}_{n,z}$ quasi-cyclic shifts. $\mathcal{G}_{n,z}$~--- cyclic group of order $z$ with generator
$\pi_0$, which acts on indexes $j$, $0 \leqslant j < nz-1$, $j = qz + r$, $0 \leqslant r < z$ in the following way:
$$
\pi_0(j) = qz + (r + 1) \mod z.
$$
$k> 0$ steps of the abstract decoder will be written in the form of a nonlinear operator $D_k: \R^{nz}\rightarrow \R^{nz}$.

If $x\in \F^{nz}$~--- some $(w, v)$-Trapping sets, $w > 0$, and the generator $\pi_0\in \mathcal{G}_{n,z}$ commutes with $D_k$, $k = 1, 2, \ldots$ then it's easy to see the error limit ~\cite{ColeHall}[Step~2] will be the same for the entire orbit  $\{\pi x: \,
\pi \in \mathcal{G}_{n,z}\}$. Moreover, automorphisms $\mathcal{G}_{n,z}$ commute with projectors $P$~(def.~\ref{defP}) Trapping sets pseudo-codewords, more precisely, the following statement holds.
\begin{statement}\label{state1}
For any divisor $z_\ast$ numbers $z$ and any $s=0, 1, \ldots, {z-1}$ the following diagram is commutative:
$$
  \begin{tikzcd}[row sep=large, column sep = large]
    \F^{nz} \arrow{r}{\pi_0^s} \arrow{d}{P_{z\rightarrow z_\ast}} & \F^{nz} \arrow{d}{P_{z\rightarrow z_\ast}} \\
    \F^{nz_\ast} \arrow{r}{\pi_0^{s\mod z_\ast}} & \F^{nz_\ast}
  \end{tikzcd}
$$
\end{statement}
\begin{proof} Follows from the definitions $P_{z\rightarrow z_\ast}$ and $\pi_0\in \mathcal{G}_{n,z}$.
\end{proof}

From previous discussions and statements  ~(\ref{state1}) it is clear that it is sufficient to leave a set of pseudo-codewords that are not pairwise equivalent to each other.

Let's choose some approximate algorithm \texttt{solve}, enumerating for the quasi-cyclic check matrix $H\in \F^{mz\times nz}$
set $X$ pseudo-code words that are not equivalent in pairs with respect to each other $\mathcal{G}_{n, z}$. Of course, there are many such algorithms, optimized in different ways, for example to:
\begin{itemize}
\item Strategy \MakeUppercase {\romannumeral 1}. get maximum number $x\in \F^{nz}$ with minimum distance $d_\epsilon^2$ up to the error limit;

\item Strategy \MakeUppercase {\romannumeral 2}. get maximum number $x\in \F^{nz}$ $x\sim (w, v)$ under restrictions $w\leqslant w_{max}$, $v\leqslant v(w)$,
where in the simplest case $v(w) = w$.
\end{itemize}

\begin{algorithm}
\label{algorithm111}
\caption{Algorithm for QC Codes TS enumeration}
\begin{algorithmic}[1]
            \Require $E(H)$, $N>0$ - number of lifted matrix $E_j$. 
            \State ~\label{alg2} Construct lifted matrices $E_j$ of the exponential matrix $E(H)$ using $B_j$ (see numerical example 1) and the corresponding parity-check matrices $H_j$, $j=0, \ldots, N-1$.
            \State  ~\label{alg3} Calculate $X_{2z,j}: = \mathtt{solve}(H_j)$.
             \State ~\label{alg4} Calculate $X_{z,j}: = P_{2z\rightarrow z} X_{2z, j}$.
             \State ~\label{alg5} Find the maximum subset of a set $\cup_{j=0}^{N-1}X_{z,j}$, consisting of pairwise nonequivalent pseudo-codewords.
\end{algorithmic}
\end{algorithm}







The Algorithm 1 computational complexity is solely determined by Step ~\ref{alg3}. The validity of the method, as per Theorem ~(\ref{th1}), is specifically established for words with class conversion $(w, v)\rightarrow (w', v')$, where $w' = w$. Notably, such cases constitute the predominant scenarios in the problems under consideration.

{\flushleft \textbf{Numeric Example 2:}} Consider the QC parity-check matrix given in ~(\ref{NumEx1.1})-(\ref{NumEx1.3}) and employ 32 matrices ($N=32$). Through 20 hours of computations utilizing four parallel threads on an Intel Xeon E5-2696 v4 processor, a total of 28,623,960 distinct $(w, v)$-Trapping Sets were founded, Strategy \MakeUppercase {\romannumeral 1}. The distribution of these words, categorized by the class $(w, v)$, is presented in tables ~\ref{tabl_ab_1} to ~\ref{tabl_ab_3}. The rows in the tables correspond to the number of variable nodes $w$, while the columns represent the count of unsatisfied checks $v$.

\begin{table}
\caption{Distribution of Trapping sets $(w, v)$} 
\centering 
\scriptsize
\tabcolsep=0.2cm
\begin{tabular}{l|l|l|l|l|l|l|l|l|l|l}
&\textbf{1}&\textbf{2}&\textbf{3}&\textbf{4}&\textbf{5}&\textbf{6}&\textbf{7}&\textbf{8}&\textbf{9}&\textbf{10}\\
\hline
\textbf{1}&0 &  &  &  &  &  &  &  &  & \\
\textbf{2}&0 & 0 &  &  &  &  &  &  &  & \\
\textbf{3}&0 & 0 & 0 &  &  &  &  &  &  & \\
\textbf{4}&0 & 0 & 0 & 1 &  &  &  &  &  & \\
\textbf{5}&0 & 0 & 0 & 2 & 33 &  &  &  &  & \\
\textbf{6}&0 & 0 & 1 & 4 & 38 & 803 &  &  &  & \\
\textbf{7}&0 & 0 & 0 & 1 & 82 & 1079 & 11109 &  &  & \\
\textbf{8}&0 & 0 & 0 & 10 & 125 & 1886 & 13722 & 82004 &  & \\
\textbf{9}&0 & 0 & 1 & 11 & 184 & 2212 & 23952 & 74786 & 257967 & \\
\textbf{10}&0 & 0 & 0 & 20 & 210 & 3122 & 20028 & 150662 & 221759 & 537954\\
\textbf{11}&0 & 0 & 0 & 14 & 249 & 2548 & 23038 & 89105 & 523548 & 434372\\
\textbf{12}&0 & 0 & 1 & 26 & 249 & 2288 & 14748 & 87080 & 253567 & 1141220\\
\textbf{13}&0 & 0 & 0 & 10 & 168 & 1638 & 10379 & 49714 & 199393 & 485167\\
\textbf{14}&0 & 0 & 0 & 14 & 153 & 1217 & 7100 & 32254 & 113767 & 314533\\
\textbf{15}&0 & 0 & 0 & 11 & 105 & 781 & 4771 & 20972 & 70765 & 183229\\
\textbf{16}&0 & 0 & 2 & 9 & 97 & 578 & 3331 & 14007 & 46158 & 116059\\
\textbf{17}&0 & 0 & 0 & 5 & 51 & 414 & 2116 & 9305 & 30284 & 75122\\
\textbf{18}&0 & 0 & 0 & 3 & 54 & 258 & 1451 & 6401 & 20480 & 49946\\
\textbf{19}&0 & 0 & 0 & 0 & 28 & 168 & 1013 & 4129 & 13555 & 33001\\
\textbf{20}&0 & 0 & 1 & 1 & 21 & 116 & 658 & 2759 & 9017 & 22248\\
\textbf{21}&0 & 0 & 0 & 1 & 6 & 73 & 394 & 1808 & 6010 & 14602\\
\textbf{22}&0 & 0 & 0 & 1 & 8 & 50 & 247 & 1184 & 3860 & 9805\\
\textbf{23}&0 & 0 & 0 & 0 & 3 & 39 & 176 & 746 & 2511 & 6446\\
\textbf{24}&0 & 0 & 0 & 0 & 0 & 24 & 91 & 499 & 1665 & 4356\\
\textbf{25}&0 & 0 & 0 & 0 & 1 & 12 & 81 & 294 & 1023 & 2759\\
\textbf{26}&0 & 0 & 0 & 0 & 0 & 7 & 51 & 196 & 672 & 1893\\
\textbf{27}&0 & 0 & 0 & 0 & 1 & 6 & 26 & 144 & 416 & 1226\\
\textbf{28}&0 & 0 & 0 & 0 & 0 & 4 & 15 & 73 & 269 & 798\\
\textbf{29}&0 & 0 & 0 & 0 & 0 & 2 & 10 & 55 & 179 & 488\\
\textbf{30}&0 & 0 & 0 & 0 & 0 & 1 & 5 & 36 & 113 & 327\\
\end{tabular}
\label{tabl_ab_1}
\end{table}

\begin{table}
\caption{Distribution of Trapping sets $(w, v)$} 
\centering 
\scriptsize
\tabcolsep=0.08cm
\begin{tabular}{l|l|l|l|l|l|l|l|l|l}
&\textbf{11}&\textbf{12}&\textbf{13}&\textbf{14}&\textbf{15}&\textbf{16}&\textbf{17}&\textbf{18}&\textbf{19}\\
\hline
\textbf{11}&679776 &  &  &  &  &  &  &  & \\
\textbf{12}&525023 & 461020 &  &  &  &  &  &  & \\
\textbf{13}&1554985 & 389114 & 258537 &  &  &  &  &  & \\
\textbf{14}&657114 & 1513059 & 256509 & 148743 &  &  &  &  & \\
\textbf{15}&388532 & 708806 & 1321139 & 151631 & 57480 &  &  &  & \\
\textbf{16}&229091 & 418952 & 672534 & 1008359 & 72773 & 31479 &  &  & \\
\textbf{17}&147227 & 250979 & 403669 & 568414 & 710717 & 38932 & 10705 &  & \\
\textbf{18}&96483 & 163431 & 244586 & 355260 & 445490 & 486945 & 16008 & 5903 & \\
\textbf{19}&64579 & 108837 & 161037 & 219709 & 296733 & 331568 & 310638 & 8197 & 2002\\
\textbf{20}&43578 & 72524 & 108737 & 146508 & 189905 & 231101 & 234722 & 201297 & 3294\\
\textbf{21}&29199 & 49859 & 73927 & 99409 & 127728 & 153153 & 176151 & 162620 & 123395\\
\textbf{22}&19921 & 33911 & 50844 & 68599 & 88005 & 104483 & 121419 & 127482 & 108790\\
\textbf{23}&13300 & 23202 & 35095 & 48044 & 60763 & 73618 & 84312 & 92364 & 91845\\
\textbf{24}&9024 & 15971 & 24386 & 33526 & 42439 & 51402 & 59559 & 65539 & 69455\\
\textbf{25}&5995 & 10959 & 17133 & 23651 & 30270 & 36385 & 42387 & 47275 & 51012\\
\textbf{26}&4229 & 7595 & 11775 & 16658 & 21354 & 25734 & 30061 & 34101 & 36835\\
\textbf{27}&2736 & 5223 & 8275 & 11831 & 15478 & 18698 & 21929 & 24431 & 27141\\
\textbf{28}&1832 & 3441 & 5762 & 8268 & 11245 & 13586 & 15902 & 18157 & 19773\\
\textbf{29}&1222 & 2310 & 3990 & 5878 & 7974 & 9774 & 11568 & 13128 & 14457\\
\textbf{30}&813 & 1537 & 2753 & 4178 & 5723 & 7115 & 8728 & 9757 & 10807\\
\end{tabular}
\label{tabl_ab_2}
\end{table}

\begin{table}
\caption{Distribution of Trapping sets $(w, v)$} 
\centering 
\scriptsize
\tabcolsep=0.12cm
\begin{tabular}{l|l|l|l|l|l|l|l|l|l|l|l}
&\textbf{20}&\textbf{21}&\textbf{22}&\textbf{23}&\textbf{24}&\textbf{25}&\textbf{26}&\textbf{27}&\textbf{28}&\textbf{29}&\textbf{30}\\
\hline
\textbf{20}& 1199 &  &  &  &  &  &  &  &  &  & \\
\textbf{21}& 1696 & 516 &  &  &  &  &  &  &  &  & \\
\textbf{22}& 77348 & 748 & 311 &  &  &  &  &  &  &  & \\
\textbf{23}& 71948 & 47293 & 429 & 194 &  &  &  &  &  &  & \\
\textbf{24}& 63799 & 46955 & 28875 & 252 & 112 &  &  &  &  &  & \\
\textbf{25}& 49967 & 44223 & 30251 & 17646 & 168 & 70 &  &  &  &  & \\
\textbf{26}& 38506 & 36280 & 29742 & 19657 & 10920 & 106 & 52 &  &  &  & \\
\textbf{27}& 28602 & 28679 & 25804 & 19810 & 12566 & 6552 & 72 & 35 &  &  & \\
\textbf{28}& 21540 & 21716 & 21008 & 17841 & 13092 & 7953 & 4066 & 41 & 32 &  & \\
\textbf{29}& 15185 & 16371 & 16191 & 14971 & 12051 & 8647 & 4949 & 2507 & 0 & 0 & \\
\textbf{30}& 11685 & 12146 & 12662 & 12214 & 10673 & 8224 & 5699 & 3093 & 1485 & 0 & 0\\
\end{tabular}
\label{tabl_ab_3}
\end{table}

\section{Quasi-Cyclic Codes Trapping Set Weighing }

This section describes techniques to expedite the estimation of the error-floor probability by leveraging the inherent characteristics of quasi-cyclic codes. The strategy involves utilizing the tabular Importance Sampling method, specifically in step 3, ~\cite{ColeHall}. In this step, the estimation of the weight of trapping sets pseudo-codewords is briefly proposed in (\cite{Richardson}, \S 4.1).

As in previous sections, $H\in \F^{mz\times nz}$, $m, n, z>0$~--- parity-check matrix defining a linear code $C(H)$, 
$D_k:\R^{nz}\rightarrow \R^{nz}$~--- decoder $k > 0$ iterations, $\mathcal{G}_{n,z}$~--- cyclic group of code automorphisms $C(H)$,
acting from the left on sets $\F^{nz}$, $\R^{nz}$ with generatrix $\pi_0$, $N:=nz$.

Probability of error in decoding a codeword of size $N$ 
with a normally distributed $N$-dimensional random noise $\xi \sim \mathcal{N}(\theta, \Sigma)$, $\Sigma = \mathrm{diag}(\sigma^2, \ldots, \sigma^2)$
equal ~\cite{ColeHall}[eq. ~5]:
\begin{equation}\label{eqPf}
P_f = \int\limits_{\R^N} I_e(y)w(y)f^\ast(y)\, dy.
\end{equation}
where  $I_e(\cdot)$~--- characteristic function of non-codewords from $\R^N$, $f^\ast(\cdot)$~--- biased density distribution,
\begin{align*}
f^\ast(y) &= 1/|\mathcal{V}|\, \sum_{x\in \mathcal{V}} f(y, x), \quad \mathcal{V}\subset \R^N, \, y\in \R^N\\
f(y, x) &= (2\pi\sigma^2)^{-N/2}\exp (-1/2\sigma^2\, \|y-x\|^2_{\R^N}), \quad x, y\in \R^N\\
w(y) &= f(y, c)/f^\ast(y), \quad y\in \R^N
\end{align*}
$f(\cdot, x)$~--- density distribution $N$-dimensional normal random variable, $f\sim \mathcal{N}(x, \Sigma)$.
$c = (1, \ldots, 1)\in \R^N$~--- encoded zero codeword. Defined $\|\cdot\|$ the Euclidean norm.

As noted in the previous section, $d_\epsilon^2(y) = d_\epsilon^2(\pi y)$, 
$\pi \in \mathcal{G}_{n,z}$, $y\in \F^N$ given that $D_k\pi = \pi D_k$, $k=1, 2, \ldots$.
It follows that it is reasonable to choose a set $\mathcal{V}$ in the form of a set of disjoint complete
$\mathcal{G}_{n,z}$-orbits of elements. We will denote these elements $\mathcal{V}_0=\{y_0^b, \ldots, y_{p-1}^b\}$,
$y_i^b\nsim y_j^b$, $i \neq j$, $y_j^b\in \R^N$. Then
\begin{gather*}
\mathcal{V} =  \mathcal{V}_0 \sqcup \pi_0 \mathcal{V}_0 \sqcup \ldots \sqcup \pi_0^{z-1} \mathcal{V}_0, \; |\mathcal{V}| = pz, \\
f^\ast(\cdot) = 1/z\, \sum_{j=0}^{z-1} f^\ast_j(\cdot), \;
\\
\text{where} \; f^\ast_j(\cdot) = 1/p\, \sum_{x\in \pi_0^j \mathcal{V}_0} f(\cdot, x), \; j = 0, \ldots, z-1.
\end{gather*}

The encoding elements \(0, 1 \in \F\) are represented by numbers \(1, -1 \in \R\) respectively. Consequently, for Trapping Sets \(x\in \F^N\), the codeword can be expressed as \(c - \mu x\in \R^N\), where \(\mu > 0\) is a method parameter. To simplify calculations, we utilize a natural embedding \(\F^N\subset \R^N\), translating elements \(\{0, 1\}\in \F\) into numbers \(\{0, 1\}\in \R\). The "basis" vectors from \(\mathcal{V}_0\) can be represented in the form \(y_j^b = c - \mu x_j^b\), where \(x_j^b\in \R^N\), \(j = 0, \ldots, p-1\), \(\{x_j^b\}_i\in \{0, 1\}\), and \(i = 0, \ldots, N-1\).

The following statement simplifies the expression $P_f$ for such a choice of set  $\mathcal{V}$.
\begin{statement}\label{statntegral}
Let's pretend that $\pi D_k = D_k \pi$, for all $\pi \in \mathcal{G}_{n,z}$, $k=1,2,\ldots$. Then
\begin{equation}\label{eqPf0}
P_f = \int\limits_{\R^N}I_e(y)w(y)f_0^\ast(y)\,dy
\end{equation}
\end{statement}
\begin{proof}
$$
P_f = 1/z\, \sum_{j=0}^{z-1}\,\int\limits_{\R^N} I_e(y)w(y)f_j^\ast(y)\, dy =: 1/z\, \sum_{j=0}^{z-1}A_j.
$$
Commutability $\pi$ and $D_k$ gives equality $I_e(\pi y) = I_e(y)$, $y\in \R^N$. It is also easy to verify $f(\pi y, x) = f(y, \pi^{-1}x), \; f^\ast(\pi y) = f^\ast(y),\;\; \forall x, y\in \R^N, \; \pi \in \mathcal{G}_{n,z}.$

Let's make a variable change $y=\pi_0^j u$ in the term $A_j$:
$$
A_j = \int\limits_{\R^N}I_e(\pi_0^j u) \frac{f(\pi_0^j u, c)}{f^\ast(\pi_0^j u)} f^\ast_j(\pi_0^j u)\, du = 
$$
$$
\int\limits_{\R^N}I_e(u) \frac{f(u, \pi_0^{-j}c)}{f^\ast(u)} f^\ast_0(u)\, du = A_0, 
$$
where used $\pi c = c$. This implies the statement.
\end{proof}

Approximation ~(\ref{eqPf0}) can be obtained using the Monte Carlo method:
\begin{equation}\label{eqPf}
P_f = 1/L\, \sum _{l=0}^{L - 1}I_e(y_l)w(y_l),
\end{equation}
where $y_l$ taken from a density distribution $f^\ast_0(\cdot)$. Calculation of value  ~(\ref{eqPf}) can be simplified further.
$$
y_l = y^b_{l\mod p} + \xi_l = c - \mu x_{l\mod p}^b + \xi_l,\,\, l = 0, 1, \ldots, \,\, \xi_l \sim \mathcal{N}(\theta, \Sigma)
$$
Let's transform the expression $w(y_l)$:
$$
w(y_l) = f(y_l, c)  /f^\ast(y_l) = zp \, \frac{R(y_l)}{S(y_l)}, \;\; l = 1, 2, \ldots.
$$
\begin{multline}\label{eqR}
R(y_l) = \exp\left( -1/2\sigma^2\, \|y_l - c\|^2_{\R^N} + \delta\right) = \\
= \exp\left( -1/2\sigma^2\, \|\mu x^b_{l\mod p} - \xi_l\|^2_{\R^N} + \delta\right)
\end{multline}
\begin{multline}\label{eqS}
S(y_l) = \sum_{j=0}^{z - 1} \sum_{x\in \pi_0^j \mathcal{V}_0} \exp \left( -1/2\sigma^2 \, \|y_l - x \| ^2_{\R^N} + \delta\right) = \\
=\sum_{j=0}^{z - 1} \sum_{k=0}^{p-1}\exp\left(-1/2\sigma^2 \| y_l - \pi_0^j y_k^b\|^2_{\R^N} + \delta\right)  = \\
=\sum_{j=0}^{z - 1} \sum_{k=0}^{p-1} \exp\left( -1/2\sigma^2 \|\mu (\pi_0^j x_k^b - x^b_{l\mod p}) + \xi_l\|^2_{\R^N} + \delta\right),
\end{multline}
$l = 0, 1, \ldots$, and $\delta$~--- a real number used in the calculation to normalize the argument of the exponent.
In the last expression ~(\ref{eqS}) let's swap the summation signs and the sum of the exponents according to $j$ from $0$ to $z-1$ let's denote 
$S_{lk}$. Then $S(y_l) = \sum_{k=0}^{p-1}S_{lk}$.

For any pair $l, k$, $k = 0, \ldots, p-1$, $l \geqslant 0$ let's select two subsets of indices $j$ namely, we get:
\begin{align}
J^1_{lk} &= \{j: \; 0 \leqslant j < z - 1, \, \pi_0^j (\supp x_k^b)\cap \supp x_{l\mod p}^b \neq \varnothing\}\label{eqJ1}\\
J^2_{lk} &= \{j: \; 0 \leqslant j < z - 1, \, \pi_0^j (\supp x_k^b)\cap \supp x_{l\mod p}^b = \varnothing\}\label{eqJ2}
\end{align}
Let's break the amount down $S_{lk}$ by two:
$$
S_{lk} = \sum_{j\in J^1_{lk}}\exp(\ldots) + \sum_{j\in J^2_{lk}}\exp(\ldots) =: S^1_{lk} + S^2_{lk}.
$$
Calculation of quantities $S^1_{lk}$ and $S^2_{lk}$ can be simplified. 
\begin{multline}\label{eqS1}
S^1_{lk}  = \exp\left(-1/2\sigma^2 \, \|\xi_l\|^2_{\R^N} + \delta\right) \times \\
 \sum _{j\in J^1_{lk}}\exp\bigg[ -\mu^2/2\sigma^2\, |(\pi_0^j(\supp x_k^b))\Delta (\supp x^b_{l\mod p})|  \\
+ \mu/\sigma^2     \bigg(\sum_{i\in \supp x^b_{l\mod p}\setminus \pi_0^j(\supp x^b_k)} \{\xi_l\}_i  \quad  \\ -
\sum_{i\in \pi_0^j(\supp x^b_k)\setminus \supp x^b_{l\mod p}} \{\xi_l\}_i \bigg)\bigg ],
\end{multline}
where $\Delta$~--- symmetric difference of sets, and

\begin{multline}\label{eqS2}
S^2_{lk} = \exp \Big( -1/2\sigma^2\, \|\xi_l\|_{\R^N} -  \mu^2/2\sigma^2\, \big(|\supp x^b_k|  \\  + |\supp x^b_{l \mod p}|\big)   +\mu/\sigma^2 \times 
 \sum_{i\in \supp x^b_{l\mod p}} \{\xi_l\}_i + \delta\Big) \\  \times \sum _{j\in J^2_{lk}}\exp \left( - 
\mu/\sigma^2 \sum_{i\in \pi_0^j (\supp _k^b)}\{\xi_l\}_i\right)
\end{multline}

Properties of equations ~(\ref{eqS1}), (\ref{eqS2}) is that the calculations are divided into index-independent $j$ and the parts that depend on it. In order to simplify direct calculations, it is enough for us to store two tables of sets:
\begin{align}
T^1_{klj}  &= \pi_0^j(\supp x_k^b)\setminus \supp x_l^b \label{eqT1}\\
T^2_{klj} &= \supp x^b_l\setminus \pi_0^j (\supp x_k^b), \label{eqT2}
\end{align}
$k, l = 0, \ldots, p-1, \, j = 0, \ldots, z-1$.  Families of sets ~(\ref{eqT1}), (\ref{eqT2}) connected by simple relations:
$\pi_0^{z-j}T^1_{k,l,j} = T^2_{l,k,z-j}$, $k,l = 0, \ldots, p - 1$, $j = 0, \ldots, z-1$. In practice, most sets on the right sides ~(\ref{eqT1}), 
(\ref{eqT2}) will not intersect.

\section{CONCLUSION}

This paper introduces an approach to enumerate and assess Trapping sets in quasi-cyclic codes with a circulant size that is not a prime number. To streamline the importance sampling step, specifically in estimating the weight of Trapping sets pseudo-codewords, a tabular quasi-cyclic method is proposed for the list of Trapping sets.

\end{document}